\pdfoutput=1
\documentclass[11pt]{article}

\newcommand{\PaperTitle}{Partial Information Decomposition When Some Sources Are Known}
\newcommand{\PaperSubtitle}{A Source Chain Rule and a No-Go Theorem}
\newcommand{\ShortTitle}{Partial Information Decomposition When Some Sources Are Known}
\newcommand{\AuthorName}{Aaron J. Gutknecht}
\newcommand{\AuthorAffiliation}{\href{https://xorlabs.org}{XORLabs}, London, United Kingdom}
\newcommand{\AuthorEmail}{contact@aarongutknecht.com}
\newcommand{\PaperDate}{}
\newcommand{\Keywords}{partial information decomposition; redundant information;
  conditioning; chain rule; local positivity; no-go theorem}

\usepackage{agpreprint}
\usepackage{placeins}

\newcommand{\B}{\mathcal{B}}
\newcommand{\PS}{\mathcal{P}}
\newcommand{\bfa}{\mathbf{a}}
\newcommand{\bfb}{\mathbf{b}}
\newcommand{\bfc}{\mathbf{c}}
\newcommand{\bfS}{\mathbf{S}}
\newcommand{\Icap}{I_\cap}
\newcommand{\Imin}{I_{\min}}
\newcommand{\Isx}{I_\cap^{\mathrm{sx}}}
\newcommand{\Rmin}{R_{\min}}
\newcommand{\Iwedge}{I_\cap^{\wedge}}
\newcommand{\ctx}{\mathbin{\Vert}}
\newcommand{\pp}[1]{\textbf{#1}}

\definecolor{RedC}{HTML}{0072B2}
\definecolor{UnqC}{HTML}{009E73}
\definecolor{SynC}{HTML}{D55E00}
\definecolor{CtxC}{RGB}{128,134,141}

\begin{document}

\makepapertitle

\begin{abstract}
Partial information decomposition (PID) splits the information that several
sources carry about a target into atoms, distinguished by which collections
of sources give access to them. We ask how such a decomposition should
behave when some of the sources are already known. The access structure of
the atoms suggests one answer: discard the atoms that the known sources
provide and regroup the rest according to which further sources unlock
them. The classical construction of conditional information suggests
another: recompute the decomposition under the conditional distributions
and average. We show that the two answers agree, atom by atom, exactly when
redundant information obeys a source chain rule, an analogue of the chain
rule of mutual information: when the same sources are added to several
collections, the enlarged collections share everything the added sources
provide and, beyond that, exactly what the original collections share once
the added sources are known. The rule divides established measures of
redundant information, so that adopting or rejecting it genuinely narrows
the choice of measure. Our main result is a no-go theorem:
for three or more sources, the source chain rule and local positivity, the
requirement that all atoms be nonnegative, cannot hold together in any
PID.
\end{abstract}

\makekeywords

\section{Introduction}
\label{sec:introduction}

Partial information decomposition (PID) asks how the information that
several source variables carry about a target is distributed among them
\citep{williams2010nonnegative}. Some of it may be available from each of
several sources (redundant information), some from a single source only
(unique information), and some only from several sources together
(synergistic information). In the mereological formulation of PID
\citep{gutknecht2021bits,gutknecht2025babel}, each such component, or
\emph{information atom}, is characterized by its access structure: the
collections of sources from which it can be obtained. This makes PID a
fine-grained language for multivariate dependence, used for instance in
neuroscience \citep{wibral2017partial,luppi2022synergistic} and in machine
learning \citep{wollstadt2023rigorous,liang2023quantifying,makkeh2025general}.

There is one atom for each possible
access structure, but this does not tell us how much information each atom
contains. If the atoms are to be the pieces of the
information that the sources provide, then the information that any
collection of sources provides should be the sum of the atoms accessible from
it. These \textit{consistency equations} tie the atoms to Shannon's mutual
information but do not determine them: for three sources, there are eighteen
atoms but only seven equations. The many proposed measures
\citep{williams2010nonnegative,Harder2013,bertschinger2014quantifying,Griffith2014_common_randomness,ince2017measuring,finn2018pointwise,rosas2020operational,makkeh2021introducing,Kolchinsky2022,gomes2024measure}
differ in which further principles they satisfy
\citep{liardi2026mathematical}. Several combinations of natural-looking
principles are known to be incompatible
\citep{Bertschinger2013_shared_info,rauh2014reconsidering,finn2018pointwise,matthias2025novel}.
Results of this kind are useful because they show which choices a measure
cannot avoid.

This paper concerns one such choice, and it is about conditioning. Suppose
that some of the sources, the \emph{context}, are already known. How should
the information that the other sources provide beyond the context be
decomposed? The atoms themselves suggest an answer. An atom that the context
does not provide becomes available when further sources are learned exactly
if it is available from these sources and the context together. The unconditional PID therefore already contains a decomposition
of the information beyond the context: discard the atoms that the context
provides, and group the others by which further sources unlock them. We call
this the \emph{contextual decomposition}. Classical information theory
suggests a different answer. As for conditional mutual information, one can
compute the PID under the conditional distribution given each value of the
context and average \citep{Fano1961}; we call this the \emph{conditional
decomposition}. Both decompose the same conditional mutual information.
Should they agree?

We show that they agree, atom by atom and for every context, if and only if
redundant information obeys a \emph{source chain rule}
(\Cref{thm:scr-atoms}). The rule concerns collections of sources that all
contain the context. What such collections share consists of all the
information in the context and, beyond that, exactly what their remaining
parts share once the context is known. For a single collection, this is the
chain rule of mutual information, which splits the information of the
collection into the information in the context and the conditional
information beyond it. The source chain rule is a genuine requirement:
established measures fall on both sides of it. It holds for three
conceptually different measures, the shared-exclusion redundancy $\Isx$
\citep{makkeh2021introducing}, the specificity--ambiguity measure $\Rmin$ of
Finn and Lizier \citep{finn2018pointwise}, and the common-information
redundancy $\Iwedge$ of Griffith et
al.~\citep{Griffith2014_common_randomness}, whereas the original measure
$\Imin$ of Williams and Beer \citep{williams2010nonnegative} violates it.
The source chain rule should be distinguished from the target chain rule
\citep{Bertschinger2013_shared_info,finn2018pointwise,makkeh2021introducing,matthias2025novel},
which conditions on part of a composite target rather than on part of the
sources. The two rules are logically independent, although under three
elementary properties the target chain rule implies the source chain rule
(\Cref{sec:tcr}).

Our main result is that the source chain rule comes at a price. For every
number of sources $n\geq3$, it is incompatible with \emph{local positivity},
the requirement that all atoms be nonnegative (\Cref{thm:incompatibility}).
The bound is sharp, since for $n\leq2$ every PID satisfies the source chain
rule. The counterexample is elementary. A context that carries no
information about the target decides which of two other sources carries it:
once the context is known, one of them is a copy of the target and the other
is constant. The source chain rule then requires that the two sources share
no information
given the context, whereas local positivity forces them to do so. The measures that satisfy the
source chain rule violate local positivity at different points of this
argument, which we make explicit.

The rest of the paper proceeds as follows. After recalling the framework
(\Cref{sec:background}), we introduce the two decompositions and show that
the source chain rule is equivalent to their agreement (\Cref{sec:scr}). We
then prove the incompatibility with local positivity
(\Cref{sec:incompatibility}), present the measures that satisfy the rule and
show where they violate local positivity (\Cref{sec:measures}), relate the
source chain rule to the target chain rule (\Cref{sec:tcr}), and discuss the
trade-off it exposes (\Cref{sec:discussion}).

\section{Background}
\label{sec:background}

We use the formulation of PID in terms of parthood distributions
\citep{gutknecht2021bits,gutknecht2025babel} and the notation of
\citep{matthias2025novel}; readers familiar with them can skip to
\Cref{sec:scr}.

Throughout, the \emph{sources} $\bfS=(S_1,\ldots,S_n)$ and the \emph{target}
$T$ are jointly distributed random variables, and logarithms are to base
two. We write
$[n]:=\{1,\ldots,n\}$, identify a collection of sources with its index set
$\bfa\subseteq[n]$, and write $S_{\bfa}:=(S_i)_{i\in\bfa}$, so that
$S_\emptyset$ is constant. Accordingly, $I(\bfa;T):=I(S_{\bfa};T)$ and
$I(\bfa;T\mid\bfc):=I(S_{\bfa};T\mid S_{\bfc})$.

\subsection{Partial information decomposition}

A PID splits the mutual information terms $I(\bfa;T)$ into atoms. What
distinguishes the atoms are their parthood relations: for every collection
$\bfa$, an atom either is or is not part of the information that $S_{\bfa}$
provides about $T$. A parthood distribution records these relations.

\begin{definition}[Parthood distribution]
\label{def:parthood}
A \emph{parthood distribution} for $n$ sources is a function
$f:\PS([n])\to\{0,1\}$ such that
\begin{enumerate}[label=\textnormal{(B\arabic*)},leftmargin=2.8em,itemsep=0.15em]
  \item $f(\emptyset)=0$,
  \item $f([n])=1$,
  \item if $\bfa\subseteq\bfb$ and $f(\bfa)=1$, then $f(\bfb)=1$.
\end{enumerate}
The set of parthood distributions for $n$ sources is denoted by $\B_n$.
\end{definition}

The value $f(\bfa)=1$ means that the atom $\Pi(f;T)$ is part of $I(\bfa;T)$;
we say that the atom is \emph{accessible} from $S_{\bfa}$. Where no
confusion can arise, we refer to the atom $\Pi(f;T)$ simply as the atom $f$.
Condition (B3)
states that what is accessible from a collection remains accessible from
every larger one. A PID assigns values to the atoms such that they add up
to the mutual information terms.

\begin{definition}[Partial information decomposition]
\label{def:pid}
A \emph{partial information decomposition} is a family of real-valued
functionals $\{\Pi(f;T)\}_{f\in\B_n}$ of the joint distribution of
$(\bfS,T)$ satisfying the \emph{consistency equations}
\begin{equation}
I(\bfa;T)=\sum_{\substack{f\in\B_n\\ f(\bfa)=1}}\Pi(f;T)
\qquad\text{for every }\bfa\subseteq[n].
\label{eq:consistency}
\end{equation}
\end{definition}

We take the PID, or the measure inducing it, to be defined for
any number of sources, as for all measures discussed in this paper, so that
it can also be applied to a subfamily of the sources.

Parthood distributions correspond one-to-one to \emph{antichains},
collections of nonempty subsets of $[n]$ none of which contains another: $f$
corresponds to the antichain of minimal collections $\bfa$ with $f(\bfa)=1$
\citep{gutknecht2021bits}. Where convenient, we also label atoms by these antichains, written
without braces and with elements separated by centered dots. For example,
$1{\cdot}23$ denotes $\{\{1\},\{2,3\}\}$, the atom that is accessible both
from $S_1$ alone and from $(S_2,S_3)$ jointly. For two sources, the four
atoms $1{\cdot}2$, $1$, $2$ and $12$ are the redundant, the two unique, and
the synergistic information.

\subsection{Redundant information}

The consistency equations do not determine the atoms: for three sources
there are eighteen atoms but only seven equations. A PID is therefore
specified through a PID-inducing measure \citep{gutknecht2025babel}, most
commonly a measure of redundant information \citep{williams2010nonnegative},
the information that each of several collections provides.

\begin{definition}[Redundant information]
\label{def:redundancy}
The \emph{redundant information} of source collections
$\bfa_1,\ldots,\bfa_m\subseteq[n]$ in a PID is
\begin{equation}
\Icap(\bfa_1,\ldots,\bfa_m;T)
:=\sum_{\substack{f\in\B_n\\ f(\bfa_1)=\cdots=f(\bfa_m)=1}}\Pi(f;T).
\label{eq:redundancy}
\end{equation}
\end{definition}

By the consistency equations, redundancy satisfies \emph{self-redundancy},
$\Icap(\bfa;T)=I(\bfa;T)$, and by (B3) it is unchanged when an argument that
contains another one is removed. It thus depends only on the antichain of
minimal arguments, and we label it like the atoms, writing for instance
$\Icap(13{\cdot}23;T)$ for $\Icap(\{1,3\},\{2,3\};T)$. Conversely, redundancy determines the
atoms. This is the basis of the original construction of Williams and Beer,
who obtain the atoms from the redundancies by Möbius inversion over their
redundancy lattice \citep{williams2010nonnegative,gutknecht2025babel}. We use
it in the following form.

\begin{lemma}[Möbius inversion]
\label{lem:mobius}
A family of real numbers $\{x(f)\}_{f\in\B_n}$ is uniquely determined by the
sums $\sum_{f:\,f(\bfa_1)=\cdots=f(\bfa_m)=1}x(f)$ over all choices of
nonempty collections $\bfa_1,\ldots,\bfa_m$.
\end{lemma}

\begin{proof}
For a parthood distribution $f$, let $\bfa_1,\ldots,\bfa_m$ be the minimal
collections with $f(\bfa)=1$. By (B3), a parthood distribution $g$ satisfies
$g(\bfa_1)=\cdots=g(\bfa_m)=1$ exactly if $g\geq f$ pointwise, so the
corresponding sum is $\sum_{g\geq f}x(g)$. Since $\B_n$ is finite, these
equations determine $x$ recursively, starting from the pointwise largest
parthood distribution: $x(f)$ is the sum belonging to $f$ minus
$\sum_{g>f}x(g)$. This recursion is the Möbius inversion of
\citep{williams2010nonnegative}.
\end{proof}

In particular, every assignment of redundancy values to antichains that
satisfies self-redundancy induces a unique PID, and we speak of the PID
induced by a redundancy measure.

\section{Conditioning a PID on some of its sources}
\label{sec:scr}

Suppose that a collection $\bfc\subseteq[n]$ of sources, the
\emph{context}, is already known. The other sources then provide the
information $I([n];T\mid\bfc)$ beyond the context, and we would like to
decompose it into atoms. Since only the remaining sources can add
information, this should be a decomposition among the remaining sources: for
three sources and the context $S_3$, for instance, into what $S_1$ and $S_2$
provide redundantly, uniquely and synergistically beyond $S_3$. This section
describes two natural ways of doing so: (1)~reading the decomposition off the
PID of all sources, and (2)~recomputing it under the conditional
distributions. Asking when the two agree leads to the source chain rule.
\Cref{fig:two-decompositions} illustrates both for three sources.

Both decompositions are indexed by the parthood distributions of the
$n-|\bfc|$ remaining sources, which we denote by $\B_{n-|\bfc|}$: functions $g$
on the collections $\bfa\subseteq[n]\setminus\bfc$ that satisfy (B1)--(B3)
with $[n]$ replaced by $[n]\setminus\bfc$. The remaining sources keep their
indices, so that for three sources and the context $S_3$ the four elements of
$\B_2$ carry the labels $1{\cdot}2$, $1$, $2$ and $12$. The first
decomposition groups the atoms of the PID of all $n$ sources; the second
recomputes the PID of the remaining sources under the conditional
distributions.

\subsection{The contextual decomposition}

Consider an atom $f$ that the context does not give access to,
$f(\bfc)=0$. An observer who already knows the context and then learns
$S_{\bfa}$, for a collection $\bfa$ of remaining sources, gains access to
this atom exactly when the combined collection $\bfa\cup\bfc$ gives access
to it. Seen from the context, the atom therefore has the \emph{access
pattern}
\begin{equation}
f_{\bfc}(\bfa):=f(\bfa\cup\bfc),\qquad\bfa\subseteq[n]\setminus\bfc.
\label{eq:access-pattern}
\end{equation}
This is a parthood distribution of the remaining sources,
$f_{\bfc}\in\B_{n-|\bfc|}$: it vanishes at $\bfa=\emptyset$ because
$f(\bfc)=0$, it equals one at $[n]\setminus\bfc$, and it inherits
monotonicity from $f$.

In terms of antichain labels, passing from $f$ to $f_{\bfc}$ is simple: erase
the context sources from every element of the label, and keep the minimal
elements. For the context $S_3$, for instance, the atom $1{\cdot}23$, which is
accessible from $S_1$ alone and from $S_2$ and $S_3$ together, becomes
$1{\cdot}2$: an observer who already knows $S_3$ gains access to it by
learning either $S_1$ or $S_2$. The atoms that the context provides itself,
$f(\bfc)=1$, are those whose label contains an element made up of context
sources only.

Atoms $f$ with the same access pattern $f_{\bfc}$ are unlocked by exactly the
same further sources: for an observer who knows the context, they play the
same role. Collecting them yields the first decomposition, a coarse-graining
of the PID of all sources.

\begin{definition}[Contextual decomposition]
\label{def:contextual}
The \emph{contextual decomposition} of the information beyond the context
$\bfc$ consists of the atoms
\begin{equation}
\Pi(g;T\ctx\bfc):=\sum_{\substack{f\in\B_n\\ f(\bfc)=0,\ f_{\bfc}=g}}\Pi(f;T),
\qquad g\in\B_{n-|\bfc|}.
\label{eq:contextual}
\end{equation}
\end{definition}

That this is indeed a decomposition of the information beyond the context is
the content of the following proposition: the contextual atoms add up to the
conditional mutual informations, just as the atoms of a PID add up to the
mutual informations.

\begin{proposition}
\label{prop:contextual-consistency}
For every collection $\bfa\subseteq[n]\setminus\bfc$,
\begin{equation}
\sum_{\substack{g\in\B_{n-|\bfc|}\\ g(\bfa)=1}}\Pi(g;T\ctx\bfc)=I(\bfa;T\mid\bfc).
\label{eq:contextual-consistency}
\end{equation}
\end{proposition}

\begin{proof}
The left-hand side is the sum of the atoms $\Pi(f;T)$ with $f(\bfc)=0$ and
$f(\bfa\cup\bfc)=1$. The atoms with $f(\bfc)=1$ also satisfy
$f(\bfa\cup\bfc)=1$, by monotonicity, and sum to $I(\bfc;T)$. Adding them
yields $I(\bfa\cup\bfc;T)$, so the left-hand side equals
$I(\bfa\cup\bfc;T)-I(\bfc;T)=I(\bfa;T\mid\bfc)$.
\end{proof}

\begin{example}[Three sources]
\label{ex:three-sources}
Let $n=3$ and $\bfc=\{3\}$. The five atoms $3$, $1{\cdot}3$, $2{\cdot}3$,
$1{\cdot}2{\cdot}3$ and $3{\cdot}12$ are accessible from $S_3$ and make up
$I(S_3;T)$. The remaining thirteen are grouped onto the four atoms $1{\cdot}2$,
$1$, $2$ and $12$, as listed in \Cref{tab:n3-grouping} and shown in
\Cref{fig:two-decompositions}A. The contextual redundancy
$\Pi(1{\cdot}2;T\ctx3)$ collects the atoms that become accessible when either
$S_1$ or $S_2$ is added to the context; the contextual unique information
$\Pi(1;T\ctx3)$ collects those that become accessible by adding $S_1$ but not
by adding $S_2$; and the contextual synergy $\Pi(12;T\ctx3)$ collects those
that become accessible only when both are added.
\end{example}

\begin{table}[htb]
\centering\small
\renewcommand{\arraystretch}{1.2}
\begin{tabular}{@{}lll@{}}
\toprule
Contextual atom & Role given $S_3$ & Atoms of the three-source PID\\
\midrule
$\Pi(1{\cdot}2;T\ctx 3)$ & redundancy & $1{\cdot}2,\ 1{\cdot}23,\ 2{\cdot}13,\ 13{\cdot}23,\ 12{\cdot}13{\cdot}23$\\
$\Pi(1;T\ctx 3)$ & unique information of $S_1$ & $1,\ 13,\ 12{\cdot}13$\\
$\Pi(2;T\ctx 3)$ & unique information of $S_2$ & $2,\ 23,\ 12{\cdot}23$\\
$\Pi(12;T\ctx 3)$ & synergy & $12,\ 123$\\
\midrule
--- & accessible from $S_3$ & $3,\ 1{\cdot}3,\ 2{\cdot}3,\ 1{\cdot}2{\cdot}3,\ 3{\cdot}12$\\
\bottomrule
\end{tabular}
\caption{The contextual decomposition for three sources and the context
$S_3$ (\Cref{ex:three-sources}). Erasing $3$ from the label of an atom that
$S_3$ does not provide, and keeping the minimal elements, yields one of the
four labels of the remaining sources $S_1$ and $S_2$. The five atoms whose
label contains the element $3$ make up $I(S_3;T)$.}
\label{tab:n3-grouping}
\end{table}

\begin{figure}[tbp]
\centering
\begingroup
\widefigure{%
\begin{tikzpicture}[font=\small,line cap=round,line join=round]
\tikzset{
  paneltitle/.style={font=\bfseries\large,text=ink,anchor=west},
  panelsub/.style={font=\small,text=muted,anchor=west},
  singlecurve/.style={draw=ink!62,line width=0.65pt},
  ctxcurve/.style={draw=ink,line width=0.95pt},
  paircurve/.style={draw=ink!80,line width=0.85pt},
  othercurve/.style={draw=ink!50,line width=0.65pt,dash pattern=on 3.2pt off 2pt},
  fullcurve/.style={draw=ink,line width=1.05pt},
  atomlabel/.style={font=\fontsize{6.4}{6.8}\selectfont,text=ink,inner sep=0.2pt},
  biglabel/.style={font=\fontsize{9.5}{10}\selectfont,text=ink,inner sep=0.3pt},
  callout/.style={font=\fontsize{8.4}{9}\selectfont,text=muted,fill=white,inner sep=1.2pt},
  leader/.style={draw=muted!78,line width=0.45pt},
  swatch/.style={minimum width=4.2mm,minimum height=3.2mm,inner sep=0pt,draw=ink!35,line width=0.4pt},
  legendtext/.style={font=\small,text=ink,anchor=west}
}

\node[paneltitle] at (0,7.95) {A\quad Contextual decomposition};
\node[panelsub] at (0,7.5) {PID of $(S_1,S_2,S_3)$; atoms grouped by their access given $S_3$};
\draw[Accent,line width=1.1pt] (0,7.2) -- (8.3,7.2);

\begin{scope}[shift={(4.15,2.45)},x=.034cm,y=.034cm]
  \coordinate (O)   at (0,0);
  \coordinate (C3)  at (0,36);
  \coordinate (C1)  at ({-18*sqrt(3)},-18);
  \coordinate (C2)  at ({ 18*sqrt(3)},-18);
  \coordinate (E13) at ({-9*sqrt(3)},9);
  \coordinate (E23) at ({ 9*sqrt(3)},9);
  \coordinate (E12) at (0,-18);
  \def\ellA{{ [rotate around={60:(E13)}] (E13) ellipse[x radius=93.5,y radius=74] }}
  \def\ellB{{ [rotate around={-60:(E23)}] (E23) ellipse[x radius=93.5,y radius=74] }}

  \begin{scope}
    \clip (O) circle[radius=105];
    \fill[SynC!26] (-110,-110) rectangle (110,110);
    \fill[UnqC!30] \ellA;
    \fill[UnqC!30] \ellB;
    \begin{scope}
      \clip \ellA;
      \fill[RedC!30] \ellB;
    \end{scope}
    \fill[CtxC!32] (C3) circle[radius=49];
  \end{scope}

  \draw[othercurve] (E12) ellipse[x radius=93.5,y radius=74];
  \draw[singlecurve] (C1) circle[radius=49];
  \draw[singlecurve] (C2) circle[radius=49];
  \draw[paircurve] \ellA;
  \draw[paircurve] \ellB;
  \draw[ctxcurve] (C3) circle[radius=49];
  \draw[fullcurve] (O) circle[radius=105];

  \node[atomlabel] at (-3.75,71) {$3$};
  \node[atomlabel] at (0,40.5) {$3{\cdot}12$};
  \node[atomlabel] at (-25,14.5) {$1{\cdot}3$};
  \node[atomlabel] at (25,14.5) {$2{\cdot}3$};
  \node[atomlabel] at (0,0) {$1{\cdot}2{\cdot}3$};
  \node[atomlabel] at (0,-29) {$1{\cdot}2$};
  \node[atomlabel] at (-35.25,-20.25) {$1{\cdot}23$};
  \node[atomlabel] at (35.25,-20.25) {$2{\cdot}13$};
  \node[atomlabel] at (18,88) {$13{\cdot}23$};
  \node[atomlabel] at (56,32) {$\star$};
  \node[atomlabel] at (-56,32) {$\star$};
  \node[atomlabel] at (0,-62.75) {$\star$};
  \node[atomlabel,anchor=west,font=\fontsize{7.2}{7.6}\selectfont]
    at (46,-112) {$\star\;=\;12{\cdot}13{\cdot}23$};
  \node[atomlabel] at (-61.5,-35.75) {$1$};
  \node[atomlabel,rotate=-42] at (-58.75,-66.5) {$12{\cdot}13$};
  \node[atomlabel] at (-72.25,41.75) {$13$};
  \node[atomlabel] at (61.5,-35.75) {$2$};
  \node[atomlabel,rotate=42] at (58.75,-66.5) {$12{\cdot}23$};
  \node[atomlabel] at (72.25,41.75) {$23$};
  \node[atomlabel] at (0,-83.5) {$12$};
  \node[atomlabel] at (-85.75,48.75) {$123$};

  \node[callout] (l3) at (-28,116) {$I_3$};
  \draw[leader] (l3.south east) -- (0,85);
  \node[callout] (l1) at (-116,-22) {$I_1$};
  \draw[leader] (l1.east) -- ({-18*sqrt(3)-49},-18);
  \node[callout] (l2) at (116,-22) {$I_2$};
  \draw[leader] (l2.west) -- ({18*sqrt(3)+49},-18);
  \node[callout] (l13) at (-112,79) {$I_{13}$};
  \draw[leader] (l13.east) -- (-65.4,65.7);
  \node[callout] (l23) at (112,79) {$I_{23}$};
  \draw[leader] (l23.west) -- (65.4,65.7);
  \node[callout] (l12) at (0,-116) {$I_{12}$};
  \draw[leader] (l12.north) -- (0,-92);
  \node[callout] (l123) at (88,114) {$I_{123}$};
  \draw[leader] (l123.south west) -- (74.25,74.25);
\end{scope}

\draw[{Stealth[length=2.6mm,width=1.9mm]}-{Stealth[length=2.6mm,width=1.9mm]},
      draw=Accent,line width=1.1pt] (8.75,2.45) -- (10.55,2.45);
\node[font=\small,text=Accent,align=center] at (9.65,3.05) {agree iff\\[-1pt]\textbf{SCR}};
\node[font=\footnotesize,text=muted] at (9.65,1.95) {Thm.~\ref{thm:scr-atoms}};

\node[paneltitle] at (10.9,7.95) {B\quad Conditional decomposition};
\node[panelsub] at (10.9,7.5) {PID of $(S_1,S_2)$ under $p(\,\cdot\mid S_3{=}z)$, averaged over $z$};
\draw[Accent,line width=1.1pt] (10.9,7.2) -- (17.2,7.2);

\begin{scope}[shift={(14.05,2.45)},x=.030cm,y=.030cm]
  \coordinate (O) at (0,0);
  \coordinate (T1) at (-28,0);
  \coordinate (T2) at (28,0);
  \begin{scope}
    \clip (O) ellipse[x radius=112,y radius=86];
    \fill[SynC!26] (-115,-90) rectangle (115,90);
    \fill[UnqC!30] (T1) circle[radius=62];
    \fill[UnqC!30] (T2) circle[radius=62];
    \begin{scope}
      \clip (T1) circle[radius=62];
      \fill[RedC!30] (T2) circle[radius=62];
    \end{scope}
  \end{scope}
  \draw[paircurve] (T1) circle[radius=62];
  \draw[paircurve] (T2) circle[radius=62];
  \draw[fullcurve] (O) ellipse[x radius=112,y radius=86];

  \node[biglabel] at (-63,0) {$1$};
  \node[biglabel] at (63,0) {$2$};
  \node[biglabel] at (0,0) {$1{\cdot}2$};
  \node[biglabel] at (0,71) {$12$};

  \node[callout] (b1) at (-110,74) {$I_{1\mid 3}$};
  \draw[leader] (b1.south east) -- ($(T1)+(136:62)$);
  \node[callout] (b2) at (110,74) {$I_{2\mid 3}$};
  \draw[leader] (b2.south west) -- ($(T2)+(44:62)$);
  \node[callout] (b12) at (0,106) {$I_{12\mid 3}$};
  \draw[leader] (b12.south) -- (0,86);
\end{scope}

\begin{scope}[shift={(0.3,-2.3)}]
  \node[swatch,fill=CtxC!32] at (0,0) {};
  \node[legendtext] at (0.25,0) {accessible from the context ($I_3$)};
  \node[swatch,fill=RedC!30] at (5.6,0) {};
  \node[legendtext] at (5.85,0) {redundancy $1{\cdot}2$};
  \node[swatch,fill=UnqC!30] at (9.0,0) {};
  \node[legendtext] at (9.25,0) {unique information $1$, $2$};
  \node[swatch,fill=SynC!26] at (13.45,0) {};
  \node[legendtext] at (13.7,0) {synergy $12$};
\end{scope}
\end{tikzpicture}%
}
\endgroup
\caption{Two decompositions of the information that $S_1$ and $S_2$ provide
about $T$ beyond the context $S_3$. \textbf{(A)}~Information diagram of the
three-source PID~\citep{williams2010nonnegative,gutknecht2021bits}. Each
curve represents a mutual-information term $I_{\mathbf a}:=I(\mathbf a;T)$ and
each cell an atom, labeled by its antichain ($1{\cdot}23$ denotes
$\{\{1\},\{2,3\}\}$; the three starred cells together form the atom
$12{\cdot}13{\cdot}23$). The five atoms inside $I_3$ are accessible from the
context alone. Each of the remaining thirteen becomes accessible by adding
$S_{\mathbf a}$ to $S_3$ exactly when it lies in $I_{\mathbf a\cup\{3\}}$, so
the atoms inside both $I_{13}$ and $I_{23}$ form the contextual redundancy,
those inside exactly one of them the contextual unique information, and
those inside neither the contextual synergy (\Cref{tab:n3-grouping}).
\textbf{(B)}~The conditional decomposition: the PID of the remaining sources
$(S_1,S_2)$ under each conditional distribution given $S_3$, averaged, with
$I_{\mathbf a\mid 3}:=I(\mathbf a;T\mid S_3)$. Both panels decompose the same
conditional mutual information $I(S_1,S_2;T\mid S_3)$, and the source chain
rule holds exactly when they agree atom by atom.}
\label{fig:two-decompositions}
\end{figure}

\subsection{The conditional decomposition}

The second way follows the classical construction of conditional
information: evaluate the same functional under the conditional
distributions and average \citep{matthias2025novel}. For a random variable
$Z$,
\begin{equation}
\begin{aligned}
\Pi(f;T\mid Z)&:=\sum_z p(z)\,\Pi^{p(\cdot\mid Z=z)}(f;T),\\
\Icap(\bfa_1,\ldots,\bfa_m;T\mid Z)&:=\sum_z p(z)\,\Icap^{p(\cdot\mid Z=z)}(\bfa_1,\ldots,\bfa_m;T),
\end{aligned}
\label{eq:conditional}
\end{equation}
where the superscript indicates evaluation under the conditional distribution
of $(\bfS,T)$ given $Z=z$. Since Möbius inversion is linear, conditional
atoms and conditional redundancies are related as their unconditional
counterparts, and the conditional atoms decompose conditional mutual
information.

We now condition on the context, $Z=S_{\bfc}$, and write ${}\mid\bfc$ for
this conditioning. Under each conditional distribution, the context sources
are constant, and the information about $T$ comes from the remaining sources.
We therefore apply the classical construction to the remaining sources:
under each conditional distribution, we compute the PID of the remaining
sources and average over the context. For $\Imin$ and the context $S_3$, for
instance, one computes the two-source $\Imin$-PID of $(S_1,S_2)$ under the
conditional distribution given $S_3=z$, and averages over $z$.

\begin{definition}[Conditional decomposition]
\label{def:conditional}
The \emph{conditional decomposition} of the information beyond the context
$\bfc$ consists of the atoms
\begin{equation}
\Pi(g;T\mid\bfc):=\sum_z p(z)\,\Pi^{p(\cdot\mid S_{\bfc}=z)}(g;T),
\qquad g\in\B_{n-|\bfc|},
\label{eq:conditional-decomposition}
\end{equation}
where $\Pi^{p(\cdot\mid S_{\bfc}=z)}(g;T)$ are the atoms of the PID of the
remaining sources under the conditional distribution given $S_{\bfc}=z$.
\end{definition}

For three sources and the context $S_3$, this is the PID of $(S_1,S_2)$ under
each conditional distribution given $S_3$, averaged
(\Cref{fig:two-decompositions}B). It is again a decomposition of the
information beyond the context, and of the same conditional mutual
informations.

\begin{proposition}
\label{prop:conditional-consistency}
For every collection $\bfa\subseteq[n]\setminus\bfc$,
\begin{equation}
\sum_{\substack{g\in\B_{n-|\bfc|}\\ g(\bfa)=1}}\Pi(g;T\mid\bfc)=I(\bfa;T\mid\bfc).
\label{eq:conditional-consistency}
\end{equation}
\end{proposition}

\begin{proof}
Under each conditional distribution, the PID of the remaining sources
satisfies the consistency equations; averaging over $z$ yields
\cref{eq:conditional-consistency}.
\end{proof}

By \Cref{prop:contextual-consistency,prop:conditional-consistency}, the
contextual and the conditional decomposition decompose the same conditional
mutual informations into atoms indexed by $\B_{n-|\bfc|}$, and the
consistency equations alone do not decide between them.

\subsection{The source chain rule}

When do the two decompositions agree? By \Cref{lem:mobius}, applied to the
remaining sources, a decomposition is determined by its redundancies, the
sums over all $g$ with $g(\bfa_1)=\cdots=g(\bfa_m)=1$. For the two
decompositions, these are as follows.

\begin{proposition}
\label{prop:redundancies}
For all nonempty $\bfa_1,\ldots,\bfa_m\subseteq[n]\setminus\bfc$,
\begin{align}
\sum_{\substack{g\in\B_{n-|\bfc|}\\ g(\bfa_1)=\cdots=g(\bfa_m)=1}}\Pi(g;T\ctx\bfc)
&=\Icap(\bfa_1\cup\bfc,\ldots,\bfa_m\cup\bfc;T)-I(\bfc;T),
\label{eq:contextual-redundancy}\\
\sum_{\substack{g\in\B_{n-|\bfc|}\\ g(\bfa_1)=\cdots=g(\bfa_m)=1}}\Pi(g;T\mid\bfc)
&=\Icap(\bfa_1,\ldots,\bfa_m;T\mid\bfc).
\label{eq:conditional-redundancy}
\end{align}
\end{proposition}

\begin{proof}
For \cref{eq:contextual-redundancy}, we first express the contextual atoms
through the atoms of the full PID and then complete the resulting sum:
\begin{align*}
\sum_{\substack{g\in\B_{n-|\bfc|}\\ g(\bfa_1)=\cdots=g(\bfa_m)=1}}\Pi(g;T\ctx\bfc)
&=\sum_{\substack{f\in\B_n,\ f(\bfc)=0\\ f(\bfa_1\cup\bfc)=\cdots=f(\bfa_m\cup\bfc)=1}}\Pi(f;T)\\
&=\sum_{\substack{f\in\B_n\\ f(\bfa_1\cup\bfc)=\cdots=f(\bfa_m\cup\bfc)=1}}\Pi(f;T)
\;-\sum_{\substack{f\in\B_n\\ f(\bfc)=1}}\Pi(f;T).
\end{align*}
The first equality inserts \cref{eq:contextual}: each atom $f$ with
$f(\bfc)=0$ contributes to exactly one contextual atom, namely $g=f_{\bfc}$,
and $f_{\bfc}(\bfa_i)=f(\bfa_i\cup\bfc)$. For the second, we add and subtract
the atoms with $f(\bfc)=1$. By (B3), each of them satisfies
$f(\bfa_i\cup\bfc)=1$ for all $i$, so adding them to the first sum yields the
sum over all atoms with $f(\bfa_i\cup\bfc)=1$ for all $i$. By
\cref{eq:redundancy} and the consistency equations, the two sums on the right
are $\Icap(\bfa_1\cup\bfc,\ldots,\bfa_m\cup\bfc;T)$ and $I(\bfc;T)$.

For \cref{eq:conditional-redundancy}, we insert
\cref{eq:conditional-decomposition} and exchange the two finite sums:
\begin{equation*}
\sum_{\substack{g\in\B_{n-|\bfc|}\\ g(\bfa_1)=\cdots=g(\bfa_m)=1}}\Pi(g;T\mid\bfc)
=\sum_zp(z)\sum_{\substack{g\in\B_{n-|\bfc|}\\ g(\bfa_1)=\cdots=g(\bfa_m)=1}}\Pi^{p(\cdot\mid S_{\bfc}=z)}(g;T).
\end{equation*}
By \cref{eq:redundancy}, applied to the PID of the remaining sources under
the conditional distribution given $S_{\bfc}=z$, the inner sum is the
redundancy $\Icap^{p(\cdot\mid S_{\bfc}=z)}(\bfa_1,\ldots,\bfa_m;T)$.
Averaging over $z$ yields the conditional redundancy \cref{eq:conditional}.
\end{proof}

In words, the contextual redundancy is what the enlarged collections
$\bfa_i\cup\bfc$ have in common, minus what the context provides by itself,
and the conditional redundancy is the redundancy recomputed under the
conditional distributions. By
\Cref{lem:mobius}, the two decompositions therefore agree exactly when
these redundancies agree for all collections, that is, when the following
identity holds. We call it the source chain rule.

\begin{definition}[Source chain rule]
\label{def:scr}
A redundancy measure satisfies the \emph{source chain rule} (\pp{SCR}) if
\begin{equation}
\Icap(\bfa_1\cup\bfc,\ldots,\bfa_m\cup\bfc;T)
=I(\bfc;T)+\Icap(\bfa_1,\ldots,\bfa_m;T\mid\bfc)
\label{eq:scr}
\end{equation}
for every $\bfc\subseteq[n]$, all nonempty
$\bfa_1,\ldots,\bfa_m\subseteq[n]\setminus\bfc$, and every joint distribution
of $(\bfS,T)$.
\end{definition}

The left-hand side is the information that the enlarged collections
$\bfa_i\cup\bfc$ have in common. Since each of them contains the context,
this includes the information carried by the context. Beyond that, an
observer who already knows the context and is then given any one of the
$S_{\bfa_i}$ learns whatever these sources provide redundantly given the
context. For a
single collection, the source chain rule is the chain rule
$I(\bfa\cup\bfc;T)=I(\bfc;T)+I(\bfa;T\mid\bfc)$ of mutual information.

\begin{example}[Three sources, continued]
\label{ex:scr-three}
For $n=3$ and the context $S_3$, the source chain rule reads
\begin{equation}
\Icap(13{\cdot}23;T)=I(S_3;T)+\Icap(1{\cdot}2;T\mid S_3).
\label{eq:scr-three}
\end{equation}
In \Cref{fig:two-decompositions}A, the left-hand side is everything inside
both $I_{13}$ and $I_{23}$. The diagram splits it into the gray atoms,
which make up $I(S_3;T)$, and the blue ones, which form the contextual
redundancy $\Pi(1{\cdot}2;T\ctx3)$: the information that becomes accessible
whichever of $S_1$ and $S_2$ is added to the context. Conceptually, the blue
part is what both sources provide once $S_3$ is known, that is, a conditional
redundancy, for which there is a standard construction
(\Cref{fig:two-decompositions}B). The source chain rule requires the two
blue regions to carry the same amount of information.
\end{example}

In summary, the two routes in the following diagram
lead to the same atoms exactly when the source chain rule holds:
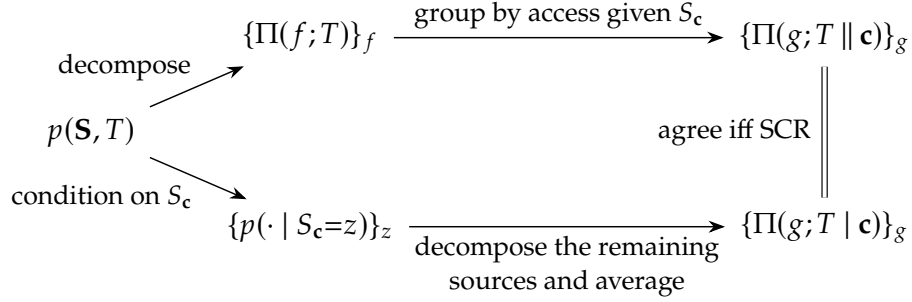
\[
\begin{tikzpicture}[x=1cm,y=1cm,baseline=(P.base),
  lab/.style={font=\small,align=center,inner sep=2pt},
  arr/.style={-{Stealth[length=2.2mm,width=1.6mm]},line width=0.55pt,
    shorten <=2pt,shorten >=2pt}]
\node (P) at (0,0) {$p(\bfS,T)$};
\node (A) at (2.9,1.25) {$\{\Pi(f;T)\}_{f}$};
\node (B) at (2.9,-1.25) {$\{p(\cdot\mid S_{\bfc}{=}z)\}_{z}$};
\node (C) at (9.7,1.25) {$\{\Pi(g;T\ctx\bfc)\}_{g}$};
\node (D) at (9.7,-1.25) {$\{\Pi(g;T\mid\bfc)\}_{g}$};
\draw[arr] (P) -- node[lab,above left] {decompose} (A);
\draw[arr] (P) -- node[lab,below left] {condition on $S_{\bfc}$} (B);
\draw[arr] (A) -- node[lab,above=1pt] {group by access given $S_{\bfc}$} (C);
\draw[arr] (B) -- node[lab,below=1pt] {decompose the remaining\\ sources and average} (D);
\draw[double,double distance=1.5pt,line width=0.45pt,shorten <=1pt,shorten >=1pt]
  (C) -- node[lab,left=3pt] {agree iff SCR} (D);
\end{tikzpicture}
\]

\begin{theorem}[When contextual equals conditional]
\label{thm:scr-atoms}
Let $\bfc\subseteq[n]$. For every joint distribution of $(\bfS,T)$, the
following statements are equivalent.
\begin{enumerate}[label=\textnormal{(\roman*)},leftmargin=2.4em,itemsep=0.2em]
  \item The source chain rule \cref{eq:scr} holds for the context $\bfc$ and
        all nonempty $\bfa_1,\ldots,\bfa_m\subseteq[n]\setminus\bfc$.
  \item The contextual and the conditional decomposition coincide:
        $\Pi(g;T\ctx\bfc)=\Pi(g;T\mid\bfc)$ for every $g\in\B_{n-|\bfc|}$.
\end{enumerate}
\end{theorem}

\begin{proof}
By \Cref{prop:redundancies}, statement (i) says that the two decompositions
have the same redundancies. By \Cref{lem:mobius}, applied to the remaining
sources, this is equivalent to (ii).
\end{proof}

In words, under the source chain rule, conditioning on a context and
recomputing the PID gives, on average, exactly the atoms of the original PID
that the context does not provide, grouped according to their access given
the context. If the rule fails, conditioning redistributes information among
redundancy, unique information and synergy in a way that the access
structure of the original atoms does not predict.

\begin{remark}[Other PID-inducing quantities]
\label{rem:other-quantities}
Statement (ii) of \Cref{thm:scr-atoms} concerns the atoms themselves, so
redundancy plays no special role in it. Every PID-inducing quantity, that is,
every sum of atoms over a condition on their parthood relations from which the
atoms can be recovered \citep{gutknecht2025babel}, has a contextual and a
conditional version, obtained by summing the contextual or the conditional
atoms over the same condition. Requiring the two versions to agree yields a
rule equivalent to the source chain rule. For union information
\citep{griffith2014quantifying,Kolchinsky2022}, the contextual version is
$I_\cup(\bfa_1\cup\bfc,\ldots,\bfa_m\cup\bfc;T)-I(\bfc;T)$, as in
\Cref{prop:redundancies}, and the rule reads
$I_\cup(\bfa_1\cup\bfc,\ldots,\bfa_m\cup\bfc;T)=I(\bfc;T)+I_\cup(\bfa_1,\ldots,\bfa_m;T\mid\bfc)$.
\end{remark}

\begin{remark}[Where the rule has content]
\label{rem:trivial}
For two sources, the source chain rule says nothing new. With the context
$S_2$, the only collection of remaining sources is $\{1\}$, and the rule
reduces to the chain rule $I(S_1,S_2;T)=I(S_2;T)+I(S_1;T\mid S_2)$ of mutual
information; with an empty context, it holds trivially. Hence every PID
satisfies the source chain rule when $n\leq2$. In general, the rule has
content only if the context is nonempty and leaves at least two other
sources. The first nontrivial instances arise for $n=3$; up to relabeling
the sources, they are the one written out in \Cref{ex:scr-three}.
\end{remark}

\section{Incompatibility with local positivity}
\label{sec:incompatibility}

A widespread requirement on a PID, beyond the consistency equations, is
that its atoms be nonnegative
\citep{williams2010nonnegative,Bertschinger2013_shared_info}.

\begin{definition}[Local positivity]
\label{def:lp}
A PID satisfies \emph{local positivity} (\pp{LP}) if all its atoms are
nonnegative, $\Pi(f;T)\geq0$, for every number of sources and every joint
distribution.
\end{definition}

This section shows that local positivity
and the source chain rule cannot hold together. On a simple distribution,
local positivity forces apart the two blue regions of
\Cref{fig:two-decompositions}, which the rule requires to carry the same
amount (\Cref{ex:scr-three}). We present this distribution first; the
theorem and its proof follow.

\subsection{A counterexample}

\begin{example}[The context decides which source carries the target]
\label{ex:witness}
Let $C$ and $T$ be independent uniformly distributed bits, and let
\begin{equation}
S_1=\max(C,T),\qquad S_2=\min(C,T),\qquad S_3=C.
\label{eq:witness}
\end{equation}
The four outcomes are equiprobable (\Cref{fig:witness}). Given $S_3=0$, the
source $S_1$ equals $T$ and $S_2$ is constant; given $S_3=1$, the roles are
reversed. The relevant mutual informations are
\begin{equation}
\begin{gathered}
I(S_1;T)=I(S_2;T)=1-\tfrac34\,h\bigl(\tfrac13\bigr)\approx0.311,\\
I(S_1,S_2;T)=\tfrac12,\qquad I(S_3;T)=0,
\end{gathered}
\label{eq:witness-mi}
\end{equation}
where $h(x)=-x\log x-(1-x)\log(1-x)$. Indeed, $S_1=0$ occurs with
probability $\tfrac14$ and implies $T=0$, while given $S_1=1$ the target
equals $1$ with probability $\tfrac23$; the same holds for $S_2$ with the
roles of $0$ and $1$ exchanged. The pairs $(S_1,S_2)=(0,0)$ and $(1,1)$
determine $T$, while $(1,0)$ occurs with probability $\tfrac12$ and leaves
$T$ uniform. Consequently, the co-information
\citep{mcgill1954multivariate,bell2003co} of $S_1$ and $S_2$ is positive:
\begin{equation}
I(S_1;T)+I(S_2;T)-I(S_1,S_2;T)=\tfrac32\bigl(1-h\bigl(\tfrac13\bigr)\bigr)\approx0.123>0.
\label{eq:coinformation}
\end{equation}
\end{example}

\begin{figure}[tbp]
\centering
\begingroup
\widefigure{%
\begin{tikzpicture}[font=\small,line cap=round,line join=round]
\tikzset{
  paneltitle/.style={font=\bfseries\large,text=ink,anchor=west},
  panelsub/.style={font=\small,text=muted,anchor=west},
  cell/.style={minimum width=10mm,minimum height=6.2mm,inner sep=0pt,anchor=center},
  quiet/.style={text=muted!80},
  carrier/.style={fill=Accent!10,draw=Accent!45,line width=0.5pt,rounded corners=2pt},
  var/.style={draw=ink!70,line width=0.7pt,rounded corners=3pt,minimum width=13mm,
              minimum height=7.5mm,fill=white,font=\normalsize},
  silentvar/.style={var,draw=muted!60,dash pattern=on 2.4pt off 1.6pt,text=muted},
  target/.style={circle,draw=Accent,line width=0.9pt,fill=Accent!10,minimum size=9mm,
                 inner sep=0pt,font=\normalsize},
  copyarrow/.style={-{Stealth[length=2.4mm,width=1.7mm]},draw=Accent,line width=1.1pt},
  worldbox/.style={draw=rule,line width=0.6pt,rounded corners=4pt},
  note/.style={font=\small,text=ink,align=center}
}

\node[paneltitle] at (0,6.25) {A\quad Joint distribution};
\node[panelsub] at (0,5.8) {four equiprobable outcomes, $S_3=C$};
\draw[Accent,line width=1.1pt] (0,5.5) -- (6.3,5.5);

\matrix (m) [matrix of math nodes,nodes={cell},column sep=1.5mm,row sep=0.6mm,
             ampersand replacement=\&]
  at (2.75,3.05) {
  C \& T \& S_1 \& S_2 \\
  0 \& 0 \& 0 \& |[quiet]| 0 \\
  0 \& 1 \& 1 \& |[quiet]| 0 \\
  1 \& 0 \& |[quiet]| 1 \& 0 \\
  1 \& 1 \& |[quiet]| 1 \& 1 \\
};
\begin{scope}[on background layer]
  \path[carrier] ([xshift=1mm]m-2-3.north west) rectangle ([xshift=-1mm]m-3-3.south east);
  \path[carrier] ([xshift=1mm]m-4-4.north west) rectangle ([xshift=-1mm]m-5-4.south east);
\end{scope}
\draw[ink!70,line width=0.6pt] (m-1-1.south west) -- (m-1-4.south east);
\draw[rule,line width=0.5pt] (m-3-1.south west) -- (m-3-4.south east);
\node[font=\footnotesize,text=muted,anchor=north] at ([yshift=-1.5mm]m-5-1.south west -| m.south)
  {$S_1=\max(C,T)$,\quad $S_2=\min(C,T)$};
\draw[decorate,decoration={brace,amplitude=4pt},ink!55,line width=0.5pt]
  ([xshift=1.5mm]m-2-4.north east) -- ([xshift=1.5mm]m-3-4.south east)
  node[midway,right=5pt,font=\footnotesize,text=ink,align=left] {$S_1=T$};
\draw[decorate,decoration={brace,amplitude=4pt},ink!55,line width=0.5pt]
  ([xshift=1.5mm]m-4-4.north east) -- ([xshift=1.5mm]m-5-4.south east)
  node[midway,right=5pt,font=\footnotesize,text=ink,align=left] {$S_2=T$};

\node[paneltitle] at (7.3,6.25) {B\quad Conditional on the context};
\node[panelsub] at (7.3,5.8) {in each world one source copies $T$, the other is constant};
\draw[Accent,line width=1.1pt] (7.3,5.5) -- (17.2,5.5);

\foreach \cx/\cval/\on/\off/\offval in {9.65/0/1/2/0, 14.85/1/2/1/1}{
  \draw[worldbox] (\cx-2.35,0.95) rectangle (\cx+2.35,5.0);
  \node[font=\normalsize,text=ink] at (\cx,4.55) {$C=\cval$};
  \node[target] (T\cval) at (\cx-1.25,2.75) {$T$};
}
\node[var] (A0) at (10.85,3.55) {$S_1$};
\node[silentvar] (B0) at (10.85,1.95) {$S_2=0$};
\draw[copyarrow] (T0) -- (A0.west) node[midway,above,sloped,font=\footnotesize,text=Accent] {copy};
\node[silentvar] (A1) at (16.05,3.55) {$S_1=1$};
\node[var] (B1) at (16.05,1.95) {$S_2$};
\draw[copyarrow] (T1) -- (B1.west) node[midway,below,sloped,font=\footnotesize,text=Accent] {copy};

\node[note] at (12.25,0.45)
  {In neither world do both sources carry information about $T$.};
\end{tikzpicture}%
}
\endgroup
\caption{The counterexample of \Cref{thm:incompatibility}. A context bit $C$
and a target bit $T$ are independent and uniform, and the sources report
their maximum and minimum. \textbf{(A)}~The four equiprobable outcomes;
shaded cells mark the source that equals $T$ within each context, gray
entries the source that is constant. \textbf{(B)}~Given $C=0$, the source
$S_1$ is a copy of $T$ and $S_2$ is constant; given $C=1$, the roles are
reversed.}
\label{fig:witness}
\end{figure}

Before turning to the formal proof, we sketch informally why this
distribution brings the two properties into conflict. Consider the two blue regions of \Cref{fig:two-decompositions} for this
distribution. Given the context, one of $S_1$ and $S_2$ is constant; it
carries no information about $T$ and, under local positivity, therefore
shares none with the other. The conditional redundancy, the blue atom in
panel~B, is thus empty. On the other hand, without the context, $S_1$ and $S_2$ together carry
less information about $T$ than the sum of what each carries alone, and
under local positivity the excess, their co-information of $0.123$~bits, is
a lower bound on what they share. For the same reason as before, none of
this shared information is also provided by $S_3$, which carries no
information about $T$. It therefore lies in the atom
$1{\cdot}2$ of panel~A: an observer who already knows $S_3$ obtains it from
$S_1$, and again from $S_2$. The blue region of panel~A, the contextual redundancy, thus holds at least
$0.123$~bits, while the blue region of panel~B is empty. Yet the source
chain rule requires the two to carry the same amount.

\subsection{The theorem}

\begin{theorem}[Incompatibility of the source chain rule and local positivity]
\label{thm:incompatibility}
For every $n\geq3$, no PID on $n$ sources satisfies both local positivity
and the source chain rule.
\end{theorem}

\begin{proof}
Suppose that a PID satisfies both properties, and consider the distribution of
\Cref{ex:witness}, together with constant sources $S_4,\ldots,S_n$ if
$n>3$; these enter none of the quantities below.

By the source chain rule \cref{eq:scr-three} with context $S_3$,
\begin{equation}
\begin{aligned}
\Icap(13{\cdot}23;T)
&=I(S_3;T)+\tfrac12\,\Icap^{p(\cdot\mid S_3=0)}(1{\cdot}2;T)\\
&\quad+\tfrac12\,\Icap^{p(\cdot\mid S_3=1)}(1{\cdot}2;T).
\end{aligned}
\end{equation}
The first term vanishes because $S_3=C$ is independent of $T$. Under
$p(\cdot\mid S_3=0)$, the source $S_2=\min(0,T)$ is constant, so its
mutual information with $T$ vanishes. In the PID of $(S_1,S_2)$ under this
distribution, this mutual information is the sum of the redundancy
$\Icap^{p(\cdot\mid S_3=0)}(1{\cdot}2;T)=\Pi^{p(\cdot\mid S_3=0)}(1{\cdot}2;T)$
and the unique information $\Pi^{p(\cdot\mid S_3=0)}(2;T)$. Both are
nonnegative by local positivity and add up to zero, so the redundancy
vanishes. The same argument
applies under $p(\cdot\mid S_3=1)$, where $S_1=\max(1,T)$ is constant. Hence
$\Icap(13{\cdot}23;T)=0$.

On the other hand, local positivity gives
\begin{equation}
\begin{aligned}
\Icap(13{\cdot}23;T)&\geq\Icap(1{\cdot}2;T)\\
&\geq I(S_1;T)+I(S_2;T)-I(S_1,S_2;T)\approx0.123>0.
\end{aligned}
\end{equation}
The first inequality holds because every atom accessible from both $S_1$ and
$S_2$ is accessible from both enlarged collections. For the second, the consistency equations give
\begin{equation*}
I(S_1;T)+I(S_2;T)-I(S_1,S_2;T)
=\sum_{f(\{1\})=f(\{2\})=1}\Pi(f;T)
\;-\sum_{\substack{f(\{1,2\})=1\\ f(\{1\})=f(\{2\})=0}}\Pi(f;T).
\end{equation*}
The first sum is the redundancy $\Icap(1{\cdot}2;T)$. The second collects
the atoms accessible from $(S_1,S_2)$ but from neither source alone. Under
local positivity, it is nonnegative, so the co-information is at most the
redundancy. This contradiction proves the theorem.
\end{proof}

Neither requirement is inconsistent on its own: $\Imin$ satisfies local
positivity \citep{williams2010nonnegative}, and \Cref{sec:measures} presents
measures that satisfy the source chain rule. For $n\leq2$, both hold
together: every PID satisfies the source chain rule there
(\Cref{rem:trivial}), $\Imin$ included.

\section{Measures that satisfy the source chain rule}
\label{sec:measures}

\Cref{thm:incompatibility} shows that the source chain rule has a price:
negative atoms. Is the rule satisfiable at all? It is, by established
measures of quite different origin. We first treat two measures that are
defined pointwise, then the common-information redundancy, and finally show
how each of the measures fares on the counterexample of
\Cref{sec:incompatibility}.

\subsection{Pointwise measures}

Several measures for discrete variables are averages of pointwise
quantities. For a realization $(s,t)$ of $(\bfS,T)$ with positive
probability, write $s_{\bfa}$ for the values of the sources in $\bfa$. A
pointwise measure assigns a \emph{pointwise redundancy}
$r(s_{\bfa_1},\ldots,s_{\bfa_m};t)$ to the source realizations
$s_{\bfa_1},\ldots,s_{\bfa_m}$ and the target realization $t$, and
\begin{equation*}
\Icap(\bfa_1,\ldots,\bfa_m;T):=\sum_{s,t}p(s,t)\,r(s_{\bfa_1},\ldots,s_{\bfa_m};t).
\end{equation*}
The pointwise redundancy depends on the joint distribution $p$; evaluated
under another distribution $q$, we write it as $r^{q}$. Two prominent
measures are of this kind. 
\begin{itemize}[leftmargin=1.6em,itemsep=0.35em]
  \item The \emph{shared-exclusion redundancy}
        $\Isx$~\citep{makkeh2021introducing} is the pointwise information
        provided by the event that at least one of the collections takes
        its observed value:
        $r(s_{\bfa_1},\ldots,s_{\bfa_m};t)=\log p\bigl(t\mid S_{\bfa_1}{=}s_{\bfa_1}\vee\cdots\vee S_{\bfa_m}{=}s_{\bfa_m}\bigr)/p(t)$.
  \item The \emph{specificity--ambiguity measure}
        $\Rmin$~\citep{finn2018pointwise} is
        $r(s_{\bfa_1},\ldots,s_{\bfa_m};t)=\min_i h(s_{\bfa_i})-\min_i h(s_{\bfa_i}\mid t)$,
        where $h(s_{\bfa}):=-\log p(s_{\bfa})$ and
        $h(s_{\bfa}\mid t):=-\log p(s_{\bfa}\mid t)$.
\end{itemize}
Both satisfy the source chain rule in a stronger, pointwise form: the
identity holds for every single realization, not only on average, just as
the chain rule of mutual information holds for pointwise mutual
information.

\begin{proposition}[Pointwise source chain rule]
\label{prop:pointwise-scr}
Let $r$ be the pointwise redundancy of $\Isx$ or of $\Rmin$. For every
context $\bfc$, all nonempty $\bfa_1,\ldots,\bfa_m\subseteq[n]\setminus\bfc$,
and every realization $(s,t)$ with positive probability,
\begin{equation}
r(s_{\bfa_1\cup\bfc},\ldots,s_{\bfa_m\cup\bfc};t)=\log\frac{p(t\mid s_{\bfc})}{p(t)}+r^{p(\cdot\mid s_{\bfc})}(s_{\bfa_1},\ldots,s_{\bfa_m};t),
\label{eq:pointwise-scr}
\end{equation}
where $p(\cdot\mid s_{\bfc})$ is the conditional distribution given
$S_{\bfc}=s_{\bfc}$. Averaging over $(s,t)$ yields the source chain rule, so
$\Isx$ and $\Rmin$ satisfy it.
\end{proposition}

\begin{proof}
For $\Isx$, let $U$ be the event that $S_{\bfa_i}=s_{\bfa_i}$ for at least
one $i$. The event that $S_{\bfa_i\cup\bfc}=s_{\bfa_i\cup\bfc}$ for at least
one $i$ is the intersection of $U$ with $\{S_{\bfc}=s_{\bfc}\}$, and
\begin{equation*}
\log\frac{p(t\mid U,s_{\bfc})}{p(t)}=\log\frac{p(t\mid s_{\bfc})}{p(t)}+\log\frac{p(t\mid U,s_{\bfc})}{p(t\mid s_{\bfc})},
\end{equation*}
where the last term is the shared-exclusion redundancy under
$p(\cdot\mid s_{\bfc})$. For $\Rmin$,
$h(s_{\bfa_i\cup\bfc})=h(s_{\bfc})+h(s_{\bfa_i}\mid s_{\bfc})$ and
$h(s_{\bfa_i\cup\bfc}\mid t)=h(s_{\bfc}\mid t)+h(s_{\bfa_i}\mid s_{\bfc},t)$
for every $i$, so the context terms can be taken out of both minima. Their
difference $h(s_{\bfc})-h(s_{\bfc}\mid t)$ is $\log p(t\mid s_{\bfc})/p(t)$,
and what remains is the specificity--ambiguity redundancy under
$p(\cdot\mid s_{\bfc})$. In both cases, averaging \cref{eq:pointwise-scr}
over $(s,t)$ yields $I(\bfc;T)$ for the first term and, after grouping the
realizations by $s_{\bfc}$, $\Icap(\bfa_1,\ldots,\bfa_m;T\mid\bfc)$ for the
second.
\end{proof}

Both measures also satisfy the
target chain rule \citep{makkeh2021introducing,finn2018pointwise}. Each is
the logarithm of a ratio between a quantity computed under the posterior
$p(\cdot\mid t)$ and the same quantity under the prior $p$, and for a
composite target $t=(t_1,t_2)$ such a logarithm telescopes into a term for
$t_1$ and a term for $t_2$ under $p(\cdot\mid t_1)$.

\subsection{Common-information redundancy}

A quite different notion of redundancy regards as shared only what is
deterministically common to the sources. A random variable is a
\emph{common part} of $X_1,\ldots,X_m$ if it is almost surely a function of
each $X_i$. There is a maximal common part $K(X_1,\ldots,X_m)$, of which
every common part is a function, unique up to relabeling
\citep{gacs1973common,Griffith2014_common_randomness},
and Griffith et al.~\citep{Griffith2014_common_randomness}
defined
$\Iwedge(\bfa_1,\ldots,\bfa_m;T):=I\bigl(K(S_{\bfa_1},\ldots,S_{\bfa_m});T\bigr)$.
For collections that all contain the context, the maximal common part
consists of the context itself together with whatever is common to the
collections within each context. This is the
key step in the proof of the following proposition.

\begin{proposition}
\label{prop:wedge-scr}
$\Iwedge$ satisfies the source chain rule.
\end{proposition}

\begin{proof}
Let $K_z$ be the maximal common part of $S_{\bfa_1},\ldots,S_{\bfa_m}$
under $p(\cdot\mid S_{\bfc}=z)$. We show that $(S_{\bfc},K_{S_{\bfc}})$ is the
maximal common part of the enlarged collections
$(S_{\bfa_1},S_{\bfc}),\ldots,(S_{\bfa_m},S_{\bfc})$. First, it is a common
part: $S_{\bfc}$ is a function of each enlarged collection, and given
$S_{\bfc}=z$, $K_z$ is a function of each $S_{\bfa_i}$. Second, every common
part $W$ of the enlarged collections is a function of it. Since the pair
contains $S_{\bfc}$, it suffices to show that, for each value $z$, $W$ is a
function of $K_z$ on the event $\{S_{\bfc}=z\}$. On this event, $W$ is a
function of each $S_{\bfa_i}$, hence a common part of
$S_{\bfa_1},\ldots,S_{\bfa_m}$ under $p(\cdot\mid S_{\bfc}=z)$, and therefore,
by maximality of $K_z$, a function of $K_z$. Hence
$\Iwedge(\bfa_1\cup\bfc,\ldots,\bfa_m\cup\bfc;T)=I(S_{\bfc},K_{S_{\bfc}};T)$,
and the chain rule of mutual information gives
\begin{equation}
\begin{aligned}
I(S_{\bfc},K_{S_{\bfc}};T)&=I(\bfc;T)+\sum_zp(z)\,I^{p(\cdot\mid S_{\bfc}=z)}(K_z;T)\\
&=I(\bfc;T)+\Iwedge(\bfa_1,\ldots,\bfa_m;T\mid\bfc).
\end{aligned}\qedhere
\end{equation}
\end{proof}

\subsection{The measures on the counterexample}

All three measures assign no shared information to a constant source. For
a constant source $S_j$, the event $S_j=s_j$ is certain, so that the event
in the definition of $\Isx$ is certain, the maximum in $\Rmin$ equals one,
and the common part in $\Iwedge$ is constant. On the counterexample, each of them
must therefore break one of the two inequalities in the proof of
\Cref{thm:incompatibility}: that redundancy does not decrease when $S_3$ is
added to both sources, or that it is at least the co-information.
\Cref{fig:escape} shows which. $\Isx$ and $\Rmin$ assign positive
redundancy to $S_1$ and $S_2$, namely $\tfrac12\log\tfrac43\approx0.21$ and
$\log\tfrac43\approx0.42$~bits, but zero redundancy once the context is
added to both sources. Their shared information thus shrinks when both
sources are supplemented by the context. $\Iwedge$ assigns zero redundancy
throughout, since $S_1$ and $S_2$ have no nonconstant common part, which is
less than their co-information. It pays with negative synergy: the atoms
accessible from $(S_1,S_2)$ but from neither source alone sum to
$-0.123$~bits.
$\Imin$, finally, satisfies local positivity \citep{williams2010nonnegative}
and therefore violates the source chain rule. On the counterexample,
$\Imin(13{\cdot}23;T)=\tfrac12$~bit, the largest value compatible with
local positivity, whereas the source chain rule requires zero.

\begin{figure}[tbp]
\centering
\begingroup
\widefigure{%
\begin{tikzpicture}[x=1cm,y=1cm,font=\small]
\tikzset{
  head/.style={font=\footnotesize,text=muted,align=center,anchor=south},
  col/.style={text=ink},
  val/.style={text=ink},
  ok/.style={text=muted},
  bad/.style={text=SynC,font=\small\bfseries,fill=SynC!12,rounded corners=2pt,
              inner xsep=6pt,inner ysep=3pt},
  name/.style={text=ink,anchor=east},
  verdict/.style={font=\footnotesize,text=ink,anchor=west},
}
\def\xZ{2.2}\def\xA{3.3}\def\xB{4.9}\def\xC{6.5}\def\xD{8.1}\def\xE{9.7}\def\xF{11.3}
\def\xV{12.2}
\node[head] at (\xA,0.75) {source\\chain rule};
\node[head] at (\xC,0.75) {adding $S_3$ does not\\decrease redundancy};
\node[head] at (\xE,0.75) {redundancy at least\\the co-information};
\draw[rule,line width=0.4pt] (\xA,0.7) -- (\xA,0.45);
\draw[rule,line width=0.4pt] (\xC,0.7) -- (\xC,0.45);
\draw[rule,line width=0.4pt] (\xE,0.7) -- (\xE,0.45);
\node[col] at (\xZ,0) {$0$};
\node[col] at (\xA,0) {$=$};
\node[col] at (\xB,0) {$I_\cap(13{\cdot}23;T)$};
\node[col] at (\xC,0) {$\geq$};
\node[col] at (\xD,0) {$I_\cap(1{\cdot}2;T)$};
\node[col] at (\xE,0) {$\geq$};
\node[col] at (\xF,0) {$0.123$};
\draw[rule,line width=0.6pt] (1.0,-0.45) -- (16.6,-0.45);
\foreach \y/\name/\b/\a/\ra/\rb/\rc/\why in {
  -1.0/{$I_{\min}$}/0.50/0.21/bad/ok/ok/{violates the source chain rule},
  -1.7/{$I_\cap^{\mathrm{sx}}$}/0/0.21/ok/bad/ok/{redundancy decreases when $S_3$ is added},
  -2.4/{$R_{\min}$}/0/0.42/ok/bad/ok/{redundancy decreases when $S_3$ is added},
  -3.1/{$I_\cap^{\wedge}$}/0/0/ok/ok/bad/{redundancy below the co-information}}{
  \node[name] at (1.6,\y) {\name};
  \node[val] at (\xZ,\y) {$0$};
  \node[val] at (\xB,\y) {$\b$};
  \node[val] at (\xD,\y) {$\a$};
  \node[val] at (\xF,\y) {$0.123$};
  \node[verdict] at (\xV,\y) {\why};
}
\node[bad] at (\xA,-1.0) {$\neq$};
\node[ok]  at (\xC,-1.0) {$\geq$};
\node[ok]  at (\xE,-1.0) {$\geq$};
\node[ok]  at (\xA,-1.7) {$=$};
\node[bad] at (\xC,-1.7) {$<$};
\node[ok]  at (\xE,-1.7) {$\geq$};
\node[ok]  at (\xA,-2.4) {$=$};
\node[bad] at (\xC,-2.4) {$<$};
\node[ok]  at (\xE,-2.4) {$\geq$};
\node[ok]  at (\xA,-3.1) {$=$};
\node[ok]  at (\xC,-3.1) {$=$};
\node[bad] at (\xE,-3.1) {$<$};
\end{tikzpicture}%
}
\endgroup
\caption{The proof of \Cref{thm:incompatibility} as a chain of three
requirements, evaluated on the counterexample of \Cref{fig:witness} (values
in bits).
All four measures assign no shared information to a constant source, so
the source chain rule forces $I_\cap(13{\cdot}23;T)=0$. Local
positivity requires that adding $S_3$ to both sources does not decrease
their redundancy, and that the redundancy of $S_1$ and $S_2$ is at least
their co-information $I(S_1;T)+I(S_2;T)-I(S_1,S_2;T)\approx0.123$. The three
requirements then contradict each other, and every measure violates at least one;
in each row, the violated requirement is highlighted. $\Imin$ violates the
source chain rule. For $\Isx$ and $\Rmin$, the redundancy decreases when
$S_3$ is added. $\Iwedge$ assigns less redundancy than the co-information, at
the price of negative synergy.}
\label{fig:escape}
\end{figure}

\FloatBarrier
\section{Relation to the target chain rule}
\label{sec:tcr}

The target chain rule
\citep{Bertschinger2013_shared_info,finn2018pointwise,makkeh2021introducing,matthias2025novel}
conditions on part of a composite target $T=(T_1,T_2)$ instead of part of
the sources:
\begin{equation}
\Icap(\bfa_1,\ldots,\bfa_m;T_1,T_2)=\Icap(\bfa_1,\ldots,\bfa_m;T_1)+\Icap(\bfa_1,\ldots,\bfa_m;T_2\mid T_1).
\label{eq:tcr}
\end{equation}
Both rules have atom-level formulations. The target chain rule is
equivalent to a chain rule for each atom \citep{matthias2025novel}, and the
source chain rule to the statement that conditioning on sources regroups the
atoms according to their access given the context (\Cref{thm:scr-atoms}).
The pointwise measures satisfy both rules, but the common-information
redundancy shows that the source chain rule does not entail the target chain
rule. For the converse, Matthias et al.\ construct, from any locally positive
base PID, a \emph{chained PID} that satisfies the target chain rule by
construction: it applies the base PID to each component of the target under
the conditional distribution given the preceding components, and it remains
locally positive \citep[Appendix]{matthias2025novel}.

\begin{proposition}
\label{prop:independence}
The source chain rule and the target chain rule are logically independent:
$\Iwedge$ satisfies the former but not the latter, and every chained PID
built from a locally positive base PID satisfies the latter but, for
$n\geq3$, not the former.
\end{proposition}

\begin{proof}
Let $S_1$ and $S_2$ be independent fair bits, $T_1=S_1\oplus S_2$ and
$T_2=S_1$. Since the sources are independent, their maximal common part is
constant, so $\Iwedge(1{\cdot}2;T_1,T_2)=\Iwedge(1{\cdot}2;T_1)=0$. Given
$T_1$, however, each source determines the other, so their maximal common
part is $S_1=T_2$ itself, and $\Iwedge(1{\cdot}2;T_2\mid T_1)=H(T_2\mid
T_1)=1$~bit. A chained PID built from a locally positive base PID
satisfies local positivity and the target chain rule
\citep{matthias2025novel}, so by \Cref{thm:incompatibility} it violates the
source chain rule for $n\geq3$.
\end{proof}

The two rules are nevertheless closely related: under three elementary
properties, the target chain rule implies the source chain rule.

\begin{proposition}
\label{prop:tcr-implies-scr}
Consider the following three properties of a redundancy measure:
\begin{enumerate}[label=\textnormal{(\roman*)},leftmargin=2.4em,itemsep=0.15em]
  \item redundancy does not change when constant sources are added to the
        collections;
  \item redundancy about a composite target does not change when its two
        components are swapped;
  \item collections that each determine the target share all of its
        information: $\Icap(\bfb_1,\ldots,\bfb_m;T)=H(T)$ whenever $T$ is a
        function of each $S_{\bfb_i}$.
\end{enumerate}
A redundancy measure that has properties (i)--(iii) and satisfies the target
chain rule also satisfies the source chain rule.
\end{proposition}

\begin{proof}
Let $\bfc\subseteq[n]$, let $\bfa_1,\ldots,\bfa_m\subseteq[n]\setminus\bfc$
be nonempty, and let $\bfb_i:=\bfa_i\cup\bfc$. We abbreviate
$\Icap(\bfb_1,\ldots,\bfb_m;\cdot)$ by $\Icap(\bfb;\cdot)$ and
$\Icap(\bfa_1,\ldots,\bfa_m;\cdot)$ by $\Icap(\bfa;\cdot)$. Expanding the
redundancy about the composite target $(S_{\bfc},T)$ by the target chain
rule in both orders gives
\begin{align*}
\Icap(\bfb;S_{\bfc},T)&=\Icap(\bfb;S_{\bfc})+\Icap(\bfb;T\mid S_{\bfc})=H(S_{\bfc})+\Icap(\bfa;T\mid\bfc),\\
\Icap(\bfb;T,S_{\bfc})&=\Icap(\bfb;T)+\Icap(\bfb;S_{\bfc}\mid T)=\Icap(\bfb;T)+H(S_{\bfc}\mid T).
\end{align*}
The second equalities hold by property (iii) for the terms with target
$S_{\bfc}$, which each $S_{\bfb_i}$ determines, also under each conditional
distribution given $T$, and by property (i) for $\Icap(\bfb;T\mid S_{\bfc})$,
since the context sources are constant under each conditional distribution
given $S_{\bfc}$. By property (ii), the left-hand sides agree, and
$H(S_{\bfc})-H(S_{\bfc}\mid T)=I(\bfc;T)$ yields the source chain rule
$\Icap(\bfb;T)=I(\bfc;T)+\Icap(\bfa;T\mid\bfc)$.
\end{proof}

The proposition explains in one stroke why $\Isx$ and $\Rmin$ satisfy the
source chain rule: both satisfy the target chain rule and, as is easily
checked, properties (i)--(iii). Under properties
(i)--(iii), the source chain rule is strictly weaker than the target chain
rule, since $\Iwedge$ has these properties and satisfies the source chain
rule, but not the target chain rule.

\Cref{tab:properties} summarizes how the measures relate to the two chain
rules and to local positivity.

\begin{table}[htb]
\centering
\renewcommand{\arraystretch}{1.18}
\begin{tabular}{@{}lccc@{}}
\toprule
& \pp{SCR} & \pp{LP} & \pp{TCR}\\
\midrule
$\Imin$~\citep{williams2010nonnegative} & no$^{a}$ & yes & no$^{b}$\\
$\Isx$~\citep{makkeh2021introducing} & yes & no & yes\\
$\Rmin$~\citep{finn2018pointwise} & yes & no & yes\\
$\Iwedge$~\citep{Griffith2014_common_randomness} & yes & no$^{c}$ & no$^{d}$\\
Chained PIDs~\citep{matthias2025novel} & no$^{a}$ & yes & yes\\
\bottomrule
\end{tabular}
\caption{The source chain rule (\pp{SCR}), local positivity (\pp{LP}) and the
target chain rule (\pp{TCR}) for $n\geq3$. Entries for \pp{SCR} follow from
\Cref{prop:pointwise-scr,prop:wedge-scr}, the others from the cited work,
except where marked: $^{a}$by \Cref{thm:incompatibility}, since the PID
satisfies \pp{LP}; $^{b}$see \citep{matthias2025novel}; $^{c}$already for two
sources \citep{liardi2026mathematical}; $^{d}$see \Cref{prop:independence}. The last row refers to every chained
PID built from a locally positive base PID, such as that of $\Imin$.}
\label{tab:properties}
\end{table}

\section{Discussion}
\label{sec:discussion}

\paragraph{The case for the source chain rule}
Conditional information has a standard construction: compute the quantity
under the conditional distributions given the context and average. At the
same time, by the very meaning of the atoms, the information that the
remaining sources provide beyond the context should be obtainable by piecing
together the atoms that become accessible once the context is known. Both
constructions are meant to capture the same thing, for instance the
information that two sources provide redundantly once a third is known, so
it seems reasonable to require that they agree. By \Cref{thm:scr-atoms}, this
requirement is the source chain rule. By \Cref{thm:incompatibility}, it
cannot be had together with local positivity. One of the two has to be given
up, and either choice leaves something to be explained.

\paragraph{Keeping the source chain rule}
A measure that satisfies the source chain rule has negative atoms, and these
need an interpretation. The counterexample shows what such an interpretation
has to account for. $\Isx$ and $\Rmin$ place the information shared by
$S_1$ and $S_2$ in the atom $1{\cdot}2{\cdot}3$, that is, they attribute it
to information that is also accessible from the context. As $I(S_3;T)=0$,
this is compensated within $I(S_3;T)$ by negative atoms
$\Pi(1{\cdot}3;T)=\Pi(2{\cdot}3;T)$, equal to $-\tfrac14$~bit for $\Isx$ and
$\tfrac14\log\tfrac38\approx-0.35$~bits for $\Rmin$. Nothing is then left to
be shared beyond the context, as the conditional decomposition requires.
But can such negative atoms be given a meaning? For pointwise measures, a natural
starting point is pointwise mutual information, which can also be negative:
an observation can lower the
probability of the actual target value, and such misinformation is a
familiar feature of pointwise information
\citep{finn2018probability,gutknecht2023information,matthias2025novel}. The
shared-exclusion redundancy is itself a pointwise mutual information, namely
that of the event that at least one source takes its observed value, so this
reading carries over to it \citep{makkeh2021introducing}. $\Rmin$ is not a
pointwise mutual information, but it is built from the specificity and the
ambiguity, the positive and negative components of pointwise information,
whose decompositions are nonnegative; its negative values arise when source
events redundantly provide misinformation about the target event
\citep{finn2018pointwise}. On the counterexample, $\Iwedge$ takes a different
route. It denies that
$S_1$ and $S_2$ share any information in the first place and pays with
negative synergy. We are not aware of an interpretation of this negativity.
As argued in \citep{matthias2025novel},
the burden of proof lies with any measure that admits negative atoms to show
that they can be interpreted.

\paragraph{Keeping local positivity}
A PID that keeps local positivity has two decompositions of the information
beyond a context, the contextual and the conditional one, and they cannot
agree for all distributions. Where they disagree, so do the contextual
and conditional redundancies, union informations and other quantities derived
from them
(\Cref{rem:other-quantities}). For $\Imin$ on the counterexample, for
instance, the contextual redundancy of $S_1$ and $S_2$ given $S_3$ is
$\tfrac12$~bit, whereas their conditional redundancy is zero: in no context
do the two sources carry information at the same time. Both decompositions
purport to describe the same thing, the information that the remaining
sources provide given that the context is already known. A PID that keeps local
positivity therefore owes an account of the difference: either only one of
them deserves this reading, and the other describes something else, or both
do, in two distinct senses that have to be told apart.

\paragraph{Chain rules and positivity}
Mutual information is nonnegative and satisfies a chain rule in each of
its arguments. For the atoms of a PID, both chain rules are at odds with local
positivity. For the source chain rule, this is \Cref{thm:incompatibility}. For
the target chain rule, it is the result of \citep{matthias2025novel}, which
needs in addition only that the PID not depend on the order in which the two
components of a composite target are listed; this holds automatically if the
target chain rule is required in both orders. The familiar behavior of mutual
information thus does not carry over to the atoms: a PID has to give up either
the nonnegativity of its atoms or the chain rules.

\section{Conclusion}
\label{sec:conclusion}

The atoms of a PID are defined by an access structure, and this structure predicts how
the decomposition should look once some of the sources are known. The source
chain rule is the requirement that this prediction agrees with what one
obtains by conditioning. It holds for the shared-exclusion,
specificity--ambiguity and common-information redundancies, and it is
incompatible with local positivity whenever there are at least three
sources. A PID can have nonnegative atoms, or atoms whose access structure correctly
predicts the effect of conditioning, but not both.

\section*{Use of AI tools}
Large language model assistants were used in preparing this paper, for
exploring arguments and counterexamples, numerically checking results,
drafting and editing text, and preparing figures. All results, proofs and
numerical values were checked by the author, who takes full responsibility
for the content.

\bibliography{references}

\end{document}